\documentclass{svjour3}                     % onecolumn (standard format)
\smartqed  % flush right qed marks, e.g. at end of proof
\usepackage{graphicx}
\usepackage{amssymb}
\begin{document}

\title{Integrable discretizations of the Dym equation
}

\titlerunning{Integrable discretizations of the Dym equation}        % if too long for running head

\author{Bao-Feng Feng, Jun-ichi Inoguchi,\\  
Kenji Kajiwara, Ken-ichi Maruno\\and Yasuhiro Ohta}

\authorrunning{B-F Feng, J. Inoguchi,\\  
K. Kajiwara, K. Maruno\\and Y. Ohta} % if too long for running head

\institute{Bao-Feng Feng and Ken-ichi Maruno\at 
Department of Mathematics, The University of Texas -- Pan American, 
              Edinburg, Texas 78539 \\
              \email{feng@utpa.edu and kmaruno@utpa.edu}           %  \\
           \and
Jun-ichi Inoguchi \at 
Department of Mathematical Sciences, Yamagata University, 
1-4-12 Kojirakawa-machi, Yamagata 990-8560, Japan\\
           Kenji Kajiwara \at
Institute of Mathematics for Industry, Kyushu University, 744 Motooka, Fukuoka 819-8581, Japan\\
Yasuhiro Ohta \at 
Department of Mathematics, Kobe University, Rokko, Kobe
  657-8501, Japan}

\date{Received: date / Accepted: date}

\maketitle

\begin{abstract}
Integrable discretizations of the complex and real Dym equations are
proposed. 
$N$-soliton solutions for both semi-discrete 
and fully discrete analogues of the complex and real Dym equations are
 also presented. 
\end{abstract}

\section{Introduction}
\label{intro}

In this article, we investigate integrable discretizations of the Dym
equation (often called the Harry Dym equation)~\cite{Kruskal,Hereman,Kadanoff,Goldstein-Petrich} 
\begin{eqnarray*}
&&{\rm the \, complex \, Dym \, equation}:\, 
r_t+r^3r_{zzz}=0\,,\quad r, z \in \mathbb{C}, t\in \mathbb{R}\,, 
\, \\
&&{\rm the \, real \, Dym \, equation}:\, r_t+r^3r_{xxx}=0 \,, \quad r,x,t \in \mathbb{R}\,.
\end{eqnarray*}
Note that the form of the complex Dym equation is same as the real Dym
equation except taking complex values $r$ and $z$ for the
complex Dym equation instead of taking real values $r$ and $x$ for the
real Dym equation, but this appears 
in various physical problems~\cite{Kadanoff,Goldstein-Petrich}.  
The Dym equation was found by Harry Dym when he was trying to transfer some
results about isospectral flows to the string equation during Martin
Kruskal's lectures~\cite{Kruskal,Hereman}. 
The Dym equation belongs to 
a class of integrable nonlinear evolution equations found by Wadati,
Konno and Ichikawa~\cite{WKI}. It is well known that the Dym equation is
transformed into the modified KdV (mKdV) equation by the hodograph
(reciprocal) transformation~\cite{Ishimori,Rogers,Kawamoto,Dmitrieva}.   

The mKdV equation
\begin{equation}
u_t \pm \frac{3}{2}u^2u_s + u_{sss}=0,\label{eqn:mKdV}
\end{equation}
where $u$ is a real function with respect to $s$ and $t$. 
Here the case of ``$+$'' sign corresponds to the focusing mKdV equation
and 
the case of ``$-$'' sign corresponds
to the defocusing mKdV equation. 
Introducing a real function $\theta(s,t)$ such that 
$u(s,t)=\frac{\partial}{\partial s}\theta(s,t)$,  
this leads to the potential mKdV equation
\begin{equation}
 \theta_t \pm \frac{1}{2}(\theta_s)^3 + \theta_{sss}=0.\label{eqn:pmKdV}
\end{equation}

For the focusing mKdV equation
\begin{equation}
u_t +\frac{3}{2}u^2u_s + u_{sss}=0,\label{eqn:mKdV-focusing}
\end{equation}
a reciprocal link between the potential focusing mKdV equation and
the complex Dym equation is deeply related to a conservation law
\begin{equation}
(e^{{\rm i}\theta})_t+\left(\frac{1}{2}(\theta_s)^2e^{{\rm i}\theta} 
+{\rm i}\theta_{ss}e^{{\rm i}\theta}
\right)_s=0\,.\label{pmKdV-focusing:conservation}
\end{equation}
Using the conserved density of this conservation law, 
we consider the hodograph (reciprocal) transformation~\cite{Kawamoto,Dmitrieva}
\begin{equation}
z(s,t)=\int_0^s e^{{\rm i}\theta(s',t)} ds'+x_0\,,\quad t'(s,t)=t\,,
\label{hodograph-complex}
\end{equation}
which leads to 
\begin{equation}
\frac{\partial}{\partial s}=
e^{{\rm i}\theta} \frac{\partial}{\partial z}\,,\qquad 
\frac{\partial}{\partial t}=
\frac{\partial}{\partial t'}+
\left(-\frac{1}{2}(\theta_s)^2e^{{\rm i}\theta} 
-{\rm i}\theta_{ss}e^{{\rm i}\theta}
\right)
\frac{\partial}{\partial z}\,.\label{hodograph-dym-focusing}
\end{equation}
Applying (\ref{hodograph-dym-focusing}) to (\ref{pmKdV-focusing:conservation})
and introducing a new dependent complex variable $r=\frac{\partial
z}{\partial s}
=e^{{\rm i}\theta}$, 
we obtain the complex Dym equation
\begin{equation}
r_{t'}+r^3r_{zzz}=0,\label{Dym-complex}
\end{equation}
which leads to 
\begin{equation}
v_{t'}+(v^{-\frac{1}{2}})_{zzz}=0,\label{Dym-complex-v}
\end{equation}
via $r=v^{-\frac{1}{2}}$.

Next, consider the defocusing modified KdV (mKdV) equation 
\begin{equation}
u_t -\frac{3}{2}u^2u_s + u_{sss}=0\,.\label{eqn:mKdV-defocusing}
\end{equation}
A reciprocal link between the potential defocusing mKdV equation and
the real Dym equation is deeply related to a conservation law
\begin{equation}
(e^{\theta})_t+\left(-\frac{1}{2}(\theta_s)^2e^{\theta} 
+\theta_{ss}e^{\theta}
\right)_s=0\,.\label{pmKdV-defocusing:conservation}
\end{equation}
Using the conserved density of this conservation law, 
we consider the hodograph (reciprocal) transformation~\cite{Kawamoto,Dmitrieva}
\begin{equation}
x(s,t)=\int_0^s e^{\theta(s',t)} ds'+x_0\,,\quad t'(s,t)=t\,,
\label{hodograph-real}
\end{equation}
which leads to 
\begin{equation}
\frac{\partial}{\partial s}=
e^{\theta} \frac{\partial}{\partial x}\,,\qquad 
\frac{\partial}{\partial t}=
\frac{\partial}{\partial t'}+
\left(\frac{1}{2}(\theta_s)^2e^{\theta} 
-\theta_{ss}e^{\theta}
\right)
\frac{\partial}{\partial x}\,.\label{hodograph-dym-defocusing}
\end{equation}
Applying (\ref{hodograph-dym-defocusing}) to (\ref{pmKdV-defocusing:conservation})
and introducing a new dependent real variable $r=\frac{\partial x}{\partial s}=e^{\theta}$, 
we obtain the real Dym equation
\begin{equation}
r_{t'}+r^3r_{xxx}=0,\label{Dym-real}
\end{equation}
which leads to 
\begin{equation}
v_{t'}+(v^{-\frac{1}{2}})_{xxx}=0,\label{Dym-real-v}
\end{equation}
via $r=v^{-\frac{1}{2}}$.

The tau-functions and bilinear equations of the complex and real Dym
equations (i.e., the focusing and defocusing mKdV equations) are given as follows: 
\begin{enumerate} 
\item The complex Dym equation is
      transformed into 
the bilinear equations of the focusing mKdV equation
\begin{eqnarray*}
&&(D_s^3+D_t)\tau\cdot \tau^*=0\,,\\
&&D_s^2\tau \cdot \tau^*=0\,, 
\end{eqnarray*}
via the dependent variable transformation
\begin{equation}
r=\left(\frac{\tau}{\tau^*}\right)^2\,, \quad \theta=\frac{2}{{\rm i}}\log \frac{\tau}{\tau^*}\,,
\end{equation}
and the hodograph (reciprocal) transformation
\begin{equation}
z(s,t)=\int_0^s e^{{\rm i}\theta(s',t)} ds'+x_0\,,\quad t'(s,t)=t\,,
\end{equation}
where $\tau^*$ is a complex conjugate of $\tau$. 
In this case, there are two types of explicit soliton solutions which
       are $N$-soliton and $M$-breather solutions. These solutions can
       be expressed by Wronskians.\\
$N$-soliton solution:
\begin{eqnarray}
&&\tau(s,t)={\rm det}\left(f_{j-1}^{(i)}\right)_{1\leq i,j\leq N}\,,
\tau^*_l(t)={\rm det}\left(f_{j}^{(i)}\right)_{1\leq i,j\leq N}\,,
\\
&&f_j^{(i)}=\alpha_ip_i^je^{p_is-p_i^3t}
+\beta_i(-p_i)^je^{-p_is+p_i^3t}\,,
\end{eqnarray}
where $p_i,\alpha_i\in \mathbb{R}$, $\beta_i=\in \sqrt{-1}\mathbb{R}$
      for $i=1,\cdots, N$. \\
$M$-breather solution:
\begin{eqnarray}
&&\tau(s,t)={\rm det}\left(f_{j-1}^{(i)}\right)_{1\leq i,j\leq N}\,,
\tau^*_l(t)={\rm det}\left(f_{j}^{(i)}\right)_{1\leq i,j\leq N}\,,
\\
&&f_j^{(i)}=\alpha_ip_i^je^{p_is-p_i^3t}
+\beta_i(-p_i)^je^{-p_is+p_i^3t}\,,
\end{eqnarray}
where $N=2M$, $p_i, \alpha_i, \beta_i\in \mathbb{C}$ for $i=1,\cdots,2M$, 
$p_{2k}=p_{2k-1}^*$, $\alpha_{2k}=\alpha_{2k-1}^*$, 
$\beta_{2k}=-\beta_{2k-1}^*$ for $k=1,\cdots, M$. 

\item  The real Dym equation is
       transformed into 
the bilinear equations of the defocusing mKdV equation
\begin{eqnarray*}
&&(D_s^3+D_t)\tau\cdot \tilde{\tau}=0\,,\\
&&D_s^2\tau \cdot \tilde{\tau}=0\,, 
\end{eqnarray*}
via the dependent variable transformation
\begin{equation}
r=\left(\frac{\tau}{\tilde{\tau}}\right)^2\,, \quad \theta=2\log \frac{\tau}{\tilde{\tau}}\,.
\end{equation}
and the hodograph (reciprocal) transformation
\begin{equation}
x(s,t)=\int_0^s e^{\theta(s',t)} ds'+x_0\,,\quad t'(s,t)=t\,.
\end{equation}
In this case, there is an $N$-cusped soliton solution which can be
       expressed by Wronskians: 
\begin{eqnarray}
&&\tau(s,t)={\rm det}\left(f_{j-1}^{(i)}\right)_{1\leq i,j\leq N}\,,
\tilde{\tau}_l(t)={\rm det}\left(f_{j}^{(i)}\right)_{1\leq i,j\leq N}\,,
\\
&&f_j^{(i)}=\alpha_ip_i^je^{p_is-p_i^3t}
+\beta_i(-p_i)^je^{-p_is+p_i^3t}\,,
\end{eqnarray}
where $p_i,\alpha_i, \beta_i\in \mathbb{R}$ for $i=1,\cdots, N$. 
\end{enumerate}

\begin{remark}
The above bilinear equations are obtained from the following bilinear
equations which belong to the modified KP hierarchy (the two-dimensional Toda lattice
hierarchy):
\begin{eqnarray}
&&(D_s^3+D_t+3D_sD_{s_2})\tau(k+1)\cdot \tau(k)=0\,,\\
&&(D_s^2-D_{s_2})\tau (k+1) \cdot \tau (k)=0\,, 
\end{eqnarray}
by imposing the conditions
\begin{equation}
\frac{\partial}{\partial s_2}\tau(k)=B\, \tau(k)\,,\quad 
\tau(k+1)=C\, {\tau^*}(k)\,, \quad B, C \in \mathbb{R}\,,  
\end{equation}
(for the real Dym equation, 
the second condition is replaced by $\tau(k+1)=C\, \tilde{\tau}(k)$)
and denoting $\tau=\tau(0)$.  
\end{remark} 

\section{Integrable semi-discrete analogues of the complex and real Dym
 equations}

\begin{lemma}\label{lemma1}
Let 
\begin{eqnarray}
&&\tau_l(t)={\rm det}\left(f_{j-1}^{(i)}\right)_{1\leq i,j\leq N}\,,
\hat{\tau}_l(t)={\rm det}\left(f_{j}^{(i)}\right)_{1\leq i,j\leq N}\,,
\\
&&f_j^{(i)}=\alpha_ip_i^j(1-\epsilon p_i)^{-l}e^{\frac{p_i}{1-\epsilon^2 p_i^2}t}
+\beta_i(-p_i)^j(1+\epsilon p_i)^{-l}e^{-\frac{p_i}{1-\epsilon^2p_i^2}t}\,.
\end{eqnarray}
These tau-functions satisfy the bilinear equations  
\begin{eqnarray}
&&
 D_t~\tau_l\cdot \hat{\tau}_l
=\frac{1}{2\epsilon}\left(\hat{\tau}_{l-1}\tau_{l+1}-\hat{\tau}_{l+1}\tau_{l-1}\right)\,,
\\
&&
 \tau_l\hat{\tau}_l=\frac{1}{2}\left(\hat{\tau}_{l-1}\tau_{l+1}+\hat{\tau}_{l+1}\tau_{l-1}\right)\,.
\end{eqnarray}
\end{lemma}
\begin{proof}
See \cite{IKMO2012-Kyushu,IKMO2012-JPA}.  
\end{proof}

\begin{theorem}\label{theorem:semi-discrete-complex-Dym}
An integrable semi-discrete
analogue of the complex Dym equation is given by 
\begin{eqnarray}
&& \frac{dr_l}{dt} =\frac{r_l}{\epsilon}
\left(\frac{r_{l+1}-r_l}{r_{l+1}+r_l}
+\frac{r_{l}-r_{l-1}}{r_{l}+r_{l-1}}\right)\,,\label{semi-discrete-complex-Dym1}\\
&&  \frac{Z_{l+1}-Z_{l}}{\epsilon} = r_l\,,\label{semi-discrete-complex-Dym2}
\end{eqnarray}
where $r_l, Z_l\in \mathbb{C}$, $t,\epsilon \in \mathbb{R}$, $l\in \mathbb{Z}$.  
The semi-discrete complex Dym equation is transformed into the bilinear equations
\begin{eqnarray}
&&
 D_t~\tau_l\cdot\tau^*_l=\frac{1}{2\epsilon}\left(\tau^*_{l-1}\tau_{l+1}-\tau_{l+1}^*\tau_{l-1}\right),
\label{semidiscrete-complex-Dym-bilinear1}
\\
&&
 \tau_l\tau^*_l=\frac{1}{2}\left(\tau^*_{l-1}\tau_{l+1}+\tau^*_{l+1}\tau_{l-1}\right),
\label{semidiscrete-complex-Dym-bilinear2}
\end{eqnarray}
via the dependent variable transformation
\begin{equation}
r_l=e^{{\rm i}\frac{\theta_{l+1}+\theta_l}{2}}=\frac{\tau_{l+1}\tau_l}{\tau_{l+1}^*\tau_l^*}\,,\quad
 \theta_l=\frac{2}{{\rm i}}\log \frac{\tau_l}{\tau_l^*}\,,
\end{equation} 
where $\tau^*_l$ is a complex conjugate of $\tau_l$. 
The $N$-soliton solution is given by 
\begin{eqnarray}
&&\tau(s,t)={\rm det}\left(f_{j-1}^{(i)}\right)_{1\leq i,j\leq N}\,,
\quad \tau^*_l(t)={\rm det}\left(f_{j}^{(i)}\right)_{1\leq i,j\leq N}\,,
\\
&&f_j^{(i)}=\alpha_ip_i^j(1-\epsilon p_i)^{-l}e^{\frac{p_i}{1-\epsilon^2 p_i^2}t}
+\beta_i(-p_i)^j(1+\epsilon p_i)^{-l}e^{-\frac{p_i}{1-\epsilon^2p_i^2}t}\,,
\end{eqnarray}
where $p_i,\alpha_i\in \mathbb{R}$, $\beta_i=\in \sqrt{-1}\mathbb{R}$
      for $i=1,\cdots, N$. \\
The $M$-breather solution is given by 
\begin{eqnarray}
&&\tau(s,t)={\rm det}\left(f_{j-1}^{(i)}\right)_{1\leq i,j\leq N}\,,
\quad \tau^*_l(t)={\rm det}\left(f_{j}^{(i)}\right)_{1\leq i,j\leq N}\,,
\\
&&f_j^{(i)}=\alpha_ip_i^j(1-\epsilon p_i)^{-l}e^{\frac{p_i}{1-\epsilon^2 p_i^2}t}
+\beta_i(-p_i)^j(1+\epsilon p_i)^{-l}e^{-\frac{p_i}{1-\epsilon^2p_i^2}t}\,,
\end{eqnarray}
where $N=2M$, $p_i, \alpha_i, \beta_i\in \mathbb{C}$ for $i=1,\cdots,2M$, 
$p_{2k}=p_{2k-1}^*$, $\alpha_{2k}=\alpha_{2k-1}^*$, 
$\beta_{2k}=-\beta_{2k-1}^*$ for $k=1,\cdots, M$. 
\end{theorem}
\begin{proof}
Here we show that the tau-functions of the bilinear equation 
(\ref{semidiscrete-complex-Dym-bilinear1}) and
 (\ref{semidiscrete-complex-Dym-bilinear2}) 
satisfy the semi-discrete Dym equation.  
Dividing (\ref{semidiscrete-complex-Dym-bilinear1}) by 
(\ref{semidiscrete-complex-Dym-bilinear2}), we obtain
\begin{equation}
\frac{d}{dt}\log \tau_{l}
-\frac{d}{dt}\log \tau^*_{l}
=\frac{1}{\epsilon}\frac{\tau_{l+1}\tau^*_{l-1}-\tau^*_{l+1}\tau_{l-1}}
{\tau_{l+1}\tau^*_{l-1}+\tau^*_{l+1}\tau_{l-1}}\,.
\end{equation} 
This can be rewritten as 
\begin{equation}
\frac{d}{dt}\log \frac{\tau_{l}}{\tau^*_{l}}
=\frac{1}{\epsilon}\frac{\frac{\tau_{l+1}\tau_l}{\tau^*_{l+1}\tau^*_l}
-\frac{\tau_{l}\tau_{l-1}}{\tau^*_{l}\tau^*_{l-1}}}
{\frac{\tau_{l+1}\tau_l}{\tau^*_{l+1}\tau^*_l}
+\frac{\tau_{l}\tau_{l-1}}{\tau^*_{l}\tau^*_{l-1}}}\,,
\end{equation}
which leads to 
\begin{equation}
\frac{d}{dt}\log \frac{\tau_{l}}{\tau^*_{l}}
=\frac{1}{\epsilon}
\frac{r_l-r_{l-1}}{r_l+r_{l-1}}
\,.\label{semi-discrete-complex-Dym-part1}
\end{equation}
Applying a shift $l \to l+1$ to (\ref{semi-discrete-complex-Dym-part1}) gives
\begin{equation}
\frac{d}{dt}\log \frac{\tau_{l+1}}{\tau^*_{l+1}}
=\frac{1}{\epsilon}
\frac{r_{l+1}-r_{l}}{r_{l+1}+r_{l}}
\,.\label{semi-discrete-complex-Dym-part2}
\end{equation}
Adding (\ref{semi-discrete-complex-Dym-part1}) 
and (\ref{semi-discrete-complex-Dym-part2}), we obtain
\begin{equation}
 \frac{d}{dt}\log \frac{\tau_{l+1}\tau_l}{\tau^*_{l+1}\tau^*_l}
=\frac{1}{\epsilon}\left(
\frac{r_{l+1}-r_{l}}{r_{l+1}+r_{l}}+
\frac{r_l-r_{l-1}}{r_l+r_{l-1}}
\right)\,,
\end{equation}
which leads to (\ref{semi-discrete-complex-Dym1}).

Equation (\ref{semi-discrete-complex-Dym2}) gives the discrete hodograph 
transformation 
\begin{equation}
Z_l=\sum_{j=0}^{l-1}\epsilon r_j+Z_0\,,
\end{equation}
which leads to the hodograph transformation (\ref{hodograph-complex}) 
in the continuous limit. 

By taking care of a complex conjugacy condition of $\tau$-functions in 
Lemma \ref{lemma1}, 
we obtain the above constraints on parameters for soliton solutions. 
\end{proof}

\begin{remark}
The above semi-discrete complex Dym equation can be written in the
 following self-adaptive moving mesh form~\cite{CH,adaptive}:
\begin{eqnarray}
&& \frac{d}{dt}(Z_{l+1}-Z_{l}) =r_l
\left(\frac{r_{l+1}-r_l}{r_{l+1}+r_l}
+\frac{r_{l}-r_{l-1}}{r_{l}+r_{l-1}}\right)\,,\label{selfadaptive-Dym1}\\
&&  r_l=\frac{Z_{l+1}-Z_{l}}{\epsilon}\,.\label{selfadaptive-Dym2}
\end{eqnarray} 
\end{remark}

We can also obtain the following theorem about an integrable
semi-discretization of the real Dym equation. 
\begin{theorem}
An integrable semi-discretization of the real Dym equation is given by 
\begin{eqnarray}
&& \frac{dr_l}{dt} =\frac{r_l}{\epsilon}
\left(\frac{r_{l+1}-r_l}{r_{l+1}+r_l}
+\frac{r_{l}-r_{l-1}}{r_{l}+r_{l-1}}\right)\,,\label{semi-discrete-real-Dym1}\\
&&  \frac{X_{l+1}-X_{l}}{\epsilon} = r_l\,,\label{semi-discrete-real-Dym2}
\end{eqnarray}
where $r_l, X_l, \epsilon \in \mathbb{R}$, $l\in \mathbb{Z}$. 
The semi-discrete real Dym equation is transformed into 
\begin{eqnarray}
&&
 D_t~\tau_l\cdot \tilde{\tau}_l=\frac{1}{2\epsilon}
\left(\tilde{\tau}_{l-1}\tau_{l+1}-\tilde{\tau}_{l+1}\tau_{l-1}\right),
\label{semidiscrete-real-Dym-bilinear1}
\\
&&
 \tau_l\tilde{\tau}_l=
\frac{1}{2}\left(\tilde{\tau}_{l-1}\tau_{l+1}+\tilde{\tau}_{l+1}\tau_{l-1}\right),
\label{semidiscrete-real-Dym-bilinear2}
\end{eqnarray}
via the dependent variable transformation 
\begin{equation}
r_l=e^{\frac{\theta_{l+1}+\theta_l}{2}}=\frac{\tau_{l+1}\tau_l}{\tilde{\tau}_{l+1}\tilde{\tau}_l}\,,\quad
 \theta_l=2\log \frac{\tau_l}{\tilde{\tau}_l}\,. 
\end{equation} 
The $N$-cusped soliton solution is given by
\begin{eqnarray}
&&\tau_l(t)={\rm det}\left(f_{j-1}^{(i)}\right)_{1\leq i,j\leq N}\,,
\quad \tilde{\tau}_l(t)={\rm det}\left(f_{j}^{(i)}\right)_{1\leq i,j\leq N}\,,
\\
&&f_j^{(i)}=\alpha_ip_i^j(1-\epsilon p_i)^{-l}e^{\frac{p_i}{1-\epsilon^2 p_i^2}t}
+\beta_i(-p_i)^j(1+\epsilon p_i)^{-l}e^{-\frac{p_i}{1-\epsilon^2p_i^2}t}\,,
\end{eqnarray}
where $p_i,\alpha_i, \beta_i\in \mathbb{R}$ for $i=1,\cdots, N$. 
\end{theorem}
\begin{proof}
The derivation of (\ref{semi-discrete-real-Dym1}) from bilinear
 equations (\ref{semidiscrete-real-Dym-bilinear1}) and 
(\ref{semidiscrete-real-Dym-bilinear2}) is same as the one in Theorem
 \ref{theorem:semi-discrete-complex-Dym}. 

Equation (\ref{semi-discrete-real-Dym2}) gives the discrete hodograph 
transformation 
\begin{equation}
X_l=\sum_{j=0}^{l-1}\epsilon r_j+X_0\,,
\end{equation}
which leads to the hodograph transformation (\ref{hodograph-real}) in 
the continuous limit. 

In the case of the semi-discrete real Dym equation, there is no 
constraint on parameters of soliton solutions in Lemma \ref{lemma1}. 
\end{proof}

The above semi-discrete complex and real Dym equations arise 
from the motion of discrete curves~\cite{FIKMO2011}. 

\section{Integrable fully discrete analogues of the complex and real Dym
 equations}

\begin{lemma}\label{lemma2}
Let 
\begin{eqnarray}
&&\tau_n^m={\rm det}\left(f_{j-1}^{(i)}\right)_{1\leq i,j\leq N}\,,
\quad \hat{\tau}_n^m={\rm det}\left(f_{j}^{(i)}\right)_{1\leq i,j\leq N}\,,
\\
&&f_j^{(i)}=\alpha_ip_i^j\prod_{n'}^{n-1}(1-a_{n'}p_i)^{-1}\prod_{m'}^{m-1}(1-b_{m'}p_i)^{-1}
\nonumber\\
&& \quad \qquad 
+\beta_i(-p_i)^j\prod_{n'}^{n-1}(1+a_{n'}p_i)^{-1}\prod_{m'}^{m-1}(1+b_{m'}p_i)^{-1}\,.
\end{eqnarray}
These tau-functions satisfy the bilinear equations  
\begin{eqnarray}
&& b_m\hat{\tau}_{n}^{m+1}\tau_{n+1}^{m} - a_n\hat{\tau}_{n+1}^{m}\tau_{n}^{m+1}
+ (a_n-b_m)\hat{\tau}_{n+1}^{m+1}\tau_n^{m}=0\,,\\
&& b_m\tau_{n}^{m+1}\hat{\tau}_{n+1}^{m} - a_n\tau_{n+1}^{m}\hat{\tau}_{n}^{m+1}
+ (a_n-b_m)\tau_{n+1}^{m+1}\hat{\tau}_n^{m}=0\,.
\end{eqnarray}
\end{lemma}
\begin{proof}
See \cite{IKMO2012-Kyushu}.  
\end{proof}

\begin{theorem}\label{theorem:discrete-complex-Dym}
An integrable fully discrete complex Dym equation is given by 
\begin{eqnarray}
&&r_n^m=Q_{n+1}^mQ_n^m\,,\label{discrete-complex-Dym1}\\
&&\frac{Q_{n+1}^{m+1}-Q_{n}^{m}}
{Q_{n+1}^{m+1}+Q_{n}^{m}}
=\frac{b_m+a_n}{b_m-a_n}~
\frac{Q_{n}^{m+1}-Q_{n+1}^{m}}
{Q_{n}^{m+1}+Q_{n+1}^{m}}\,,\label{discrete-complex-Dym2}\\
&&Z_{n+1}^m-Z_{n}^m=a_nr_n^m\,,\label{discrete-complex-Dym3}
\end{eqnarray}
where $r_n^m, Q_n^m, Z_n^m \in \mathbb{C}$, $a_n, b_m \in \mathbb{R}$, 
$m,n\in \mathbb{Z}$. 
This is transformed into 
\begin{eqnarray}
&& b_m\tau^*{}_{n}^{m+1}\tau_{n+1}^{m} - a_n\tau^*{}_{n+1}^{m}\tau_{n}^{m+1}
+ (a_n-b_m)\tau^*{}_{n+1}^{m+1}\tau_n^{m}=0\,,
\label{discrete-complex-Dym-bilinear1}\\
&& b_m\tau_{n}^{m+1}\tau^*{}_{n+1}^{m} - a_n\tau_{n+1}^{m}\tau^*{}_{n}^{m+1}
+ (a_n-b_m)\tau_{n+1}^{m+1}\tau^*{}_n^{m}=0\,, 
\label{discrete-complex-Dym-bilinear2}
\end{eqnarray}
via
\begin{equation}
r_n^m=e^{{\rm i}\frac{\theta_{n+1}^m+\theta_n^m}{2}}
=\frac{\tau_{n+1}^m\tau_n^m}{{\tau^*}_{n+1}^m{\tau^*}_n^m}\,,
\quad Q_n^m= \frac{\tau_n^m}{{\tau^*}_n^m}
\,,\quad
\theta_n=\frac{2}{{\rm i}}\log \frac{\tau_n^m}{{\tau^*}_n^m} \,,
\end{equation}
where ${\tau^*}_n^m$ is a complex conjugate of $\tau_n^m$. 
The $N$-soliton solution is given by 
\begin{eqnarray}
&&\tau_n^m={\rm det}\left(f_{j-1}^{(i)}\right)_{1\leq i,j\leq N}\,,
\quad {\tau^*}_n^m={\rm det}\left(f_{j}^{(i)}\right)_{1\leq i,j\leq N}\,,
\\
&&f_j^{(i)}=\alpha_ip_i^j\prod_{n'}^{n-1}(1-a_{n'}p_i)^{-1}\prod_{m'}^{m-1}(1-b_{m'}p_i)^{-1}
\nonumber\\
&& \quad \qquad 
+\beta_i(-p_i)^j\prod_{n'}^{n-1}(1+a_{n'}p_i)^{-1}\prod_{m'}^{m-1}(1+b_{m'}p_i)^{-1}\,,
\end{eqnarray}
where $p_i,\alpha_i\in \mathbb{R}$, $\beta_i \in \sqrt{-1}\mathbb{R}$
      for $i=1,\cdots, N$. \\
The $M$-breather solution is given by 
\begin{eqnarray}
&&\tau_n^m={\rm det}\left(f_{j-1}^{(i)}\right)_{1\leq i,j\leq N}\,,
\quad {\tau^*}_n^m={\rm det}\left(f_{j}^{(i)}\right)_{1\leq i,j\leq N}\,,
\\
&&f_j^{(i)}=\alpha_ip_i^j\prod_{n'}^{n-1}(1-a_{n'}p_i)^{-1}\prod_{m'}^{m-1}(1-b_{m'}p_i)^{-1}
\nonumber\\
&& \quad \qquad 
+\beta_i(-p_i)^j\prod_{n'}^{n-1}(1+a_{n'}p_i)^{-1}\prod_{m'}^{m-1}(1+b_{m'}p_i)^{-1}\,,
\end{eqnarray}
where $N=2M$, $p_i, \alpha_i, \beta_i\in \mathbb{C}$ for $i=1,\cdots,2M$, 
$p_{2k}=p_{2k-1}^*$, $\alpha_{2k}=\alpha_{2k-1}^*$, 
$\beta_{2k}=-\beta_{2k-1}^*$ for $k=1,\cdots, M$. 
\end{theorem}
\begin{proof}
We show that the tau-functions of the bilinear equation 
(\ref{discrete-complex-Dym-bilinear1}) and
 (\ref{discrete-complex-Dym-bilinear2}) 
satisfy the fully discrete Dym equation.  
Subtracting (\ref{discrete-complex-Dym-bilinear1}) from 
(\ref{discrete-complex-Dym-bilinear2}), we obtain 
\begin{equation}
\tau_{n+1}^{m+1}\tau^*{}_n^{m}-\tau^*{}_{n+1}^{m+1}\tau_n^{m}
=\frac{b_m+a_n}{b_m-a_n}
(\tau_{n}^{m+1}\tau^*{}_{n+1}^{m}-\tau_{n+1}^{m}\tau^*{}_n^{m+1})\,. 
\label{discrete-complex-Dym-bilinear-combine1}
\end{equation}
Adding (\ref{discrete-complex-Dym-bilinear1}) 
and (\ref{discrete-complex-Dym-bilinear2}), we obtain 
\begin{equation}
\tau_{n+1}^{m+1}\tau^*{}_n^{m}+\tau^*{}_{n+1}^{m+1}\tau_n^{m}
=\tau_{n}^{m+1}\tau^*{}_{n+1}^{m}+\tau_{n+1}^{m}\tau^*{}_n^{m+1}\,. 
\label{discrete-complex-Dym-bilinear-combine2}
\end{equation}
Dividing (\ref{discrete-complex-Dym-bilinear-combine1}) by 
(\ref{discrete-complex-Dym-bilinear-combine2}), we obtain
\begin{equation}
\frac{\tau_{n+1}^{m+1}\tau^*{}_n^{m}-\tau^*{}_{n+1}^{m+1}\tau_n^{m}}
{\tau_{n+1}^{m+1}\tau^*{}_n^{m}+\tau^*{}_{n+1}^{m+1}\tau_n^{m}}=
\frac{b_m+a_n}{b_m-a_n}\frac{\tau_{n}^{m+1}\tau^*{}_{n+1}^{m}-\tau_{n+1}^{m}\tau^*{}_n^{m+1}}
{\tau_{n}^{m+1}\tau^*{}_{n+1}^{m}+\tau_{n+1}^{m}\tau^*{}_n^{m+1}}\,,
\end{equation}
which leads to 
\begin{equation}
\frac{
\frac{\tau_{n+1}^{m+1}}{\tau^*{}_{n+1}^{m+1}}-
\frac{\tau_n^{m}}{\tau^*{}_n^{m}}
}
{
\frac{\tau_{n+1}^{m+1}}{\tau^*{}_{n+1}^{m+1}}+
\frac{\tau_n^{m}}{\tau^*{}_n^{m}}
}=
\frac{b_m+a_n}{b_m-a_n}
\frac{
\frac{\tau_{n}^{m+1}}{\tau^*{}_n^{m+1}}
-\frac{\tau_{n+1}^{m}}{\tau^*{}_{n+1}^{m}}
}
{
\frac{\tau_{n}^{m+1}}{\tau^*{}_n^{m+1}}
+\frac{\tau_{n+1}^{m}}{\tau^*{}_{n+1}^{m}}
}\,. 
\end{equation}
This gives (\ref{discrete-complex-Dym2}). 

Equation (\ref{discrete-complex-Dym3}) gives the discrete hodograph 
transformation 
\begin{equation}
Z_n^m=\sum_{j=0}^{n-1}a_j r_j^m+Z_0^m\,,
\end{equation}
which leads to the hodograph transformation (\ref{hodograph-complex}) in
 the continuous limit. 

By taking care of a complex conjugacy condition of $\tau$-functions in Lemma \ref{lemma2}, 
we obtain the above constraints on parameters fot soliton solutions. 
\end{proof}

We can also obtain the following theorem about an integrable fully 
discrete real Dym equation.  
\begin{theorem}\label{theorem:discrete-real-Dym}
An integrable fully discrete real Dym equation is given by 
\begin{eqnarray}
&&r_n^m=Q_{n+1}^mQ_n^m\,,\label{discrete-real-Dym1}\\
&&\frac{Q_{n+1}^{m+1}-Q_{n}^{m}}
{Q_{n+1}^{m+1}+Q_{n}^{m}}
=\frac{b_m+a_n}{b_m-a_n}~
\frac{Q_{n}^{m+1}-Q_{n+1}^{m}}
{Q_{n}^{m+1}+Q_{n+1}^{m}}\,,\label{discrete-real-Dym2}\\
&&X_{n+1}^m-X_{n}^m=a_nr_n^m\,,\label{discrete-real-Dym3}
\end{eqnarray}
where $r_n^m, Q_n^m, X_n^m, a_n, b_m \in \mathbb{R}$, 
$m,n\in \mathbb{Z}$. 
This is transformed into 
\begin{eqnarray}
&& b_m\tilde{\tau}_{n}^{m+1}\tau_{n+1}^{m} - a_n\tilde{\tau}_{n+1}^{m}\tau_{n}^{m+1}
+ (a_n-b_m)\tilde{\tau}_{n+1}^{m+1}\tau_n^{m}=0\,,\label{discrete-real-Dym-bilinear1}\\
&& b_m\tau_{n}^{m+1}\tilde{\tau}_{n+1}^{m} - a_n\tau_{n+1}^{m}\tilde{\tau}_{n}^{m+1}
+ (a_n-b_m)\tau_{n+1}^{m+1}\tilde{\tau}_n^{m}=0\,, \label{discrete-real-Dym-bilinear2}
\end{eqnarray}
via
\begin{equation}
r_n^m=e^{\frac{\theta_{n+1}^m+\theta_n^m}{2}}=\frac{\tau_{n+1}^m\tau_n^m}
{\tilde{\tau}_{n+1}^m\tilde{\tau}_n^m}\,,
\quad Q_n^m= \frac{\tau_n^m}{\tilde{\tau}_n^m}
\,,\quad
 \theta_n=2\log \frac{\tau_n^m}{\tilde{\tau}_n^m}\,. 
\end{equation}
The $N$-cusped soliton solution is given by
\begin{eqnarray}
&&\tau_n^m={\rm det}\left(f_{j-1}^{(i)}\right)_{1\leq i,j\leq N}\,,
\quad \tilde{\tau}_n^m={\rm det}\left(f_{j}^{(i)}\right)_{1\leq i,j\leq N}\,,
\\
&&f_j^{(i)}=\alpha_ip_i^j\prod_{n'}^{n-1}(1-a_{n'}p_i)^{-1}\prod_{m'}^{m-1}(1-b_{m'}p_i)^{-1}
\nonumber\\
&& \quad \qquad 
+\beta_i(-p_i)^j\prod_{n'}^{n-1}(1+a_{n'}p_i)^{-1}\prod_{m'}^{m-1}(1+b_{m'}p_i)^{-1}\,,
\end{eqnarray}
where $p_i,\alpha_i, \beta_i\in \mathbb{R}$ for $i=1,\cdots, N$. 
\end{theorem}

\begin{proof}
 The derivation of the fully discrete real Dym equation from bilinear
 equations (\ref{discrete-real-Dym-bilinear1}) and 
(\ref{discrete-real-Dym-bilinear2}) is same as the one in Theorem
 \ref{theorem:discrete-complex-Dym}. 
In the case of the fully discrete real Dym equation, 
there is no constraint on parameters of soliton solutions in Lemma \ref{lemma2}. 
\end{proof}

\begin{remark}
Equation (\ref{discrete-real-Dym3}) gives the discrete hodograph 
transformation 
\begin{equation}
X_n^m=\sum_{j=0}^{n-1}a_j r_j^m+X_0^m\,,
\end{equation}
which leads to the hodograph transformation (\ref{hodograph-real}) in
 the continuous limit. 
\end{remark}

The above fully discrete complex and real Dym equations arise 
from the motion of discrete curves~\cite{FIKMO2011}. 

Based on bilinear equations, 
we also obtain another form of fully discrete complex and real Dym equations. 
\begin{theorem}
An integrable fully discrete complex Dym equation is given by 
\begin{eqnarray}
&&\frac{\sqrt{r_{n+1}^{m+1}}-\sqrt{r_{n}^{m}}}
{\sqrt{r_{n+1}^{m+1}}+\sqrt{r_{n}^{m}}}
=\frac{b_m+a_n}{b_m-a_n}~
\frac{\sqrt{r_{n}^{m+1}}-\sqrt{r_{n+1}^{m}}}
{\sqrt{r_{n}^{m+1}}+\sqrt{r_{n+1}^{m}}}\,,\label{discrete-complex-Dym1:2nd}\\
&&Z_{n+1}^m-Z_{n}^m=a_n\sqrt{r_{n+1}^mr_n^m}\,,\label{discrete-complex-Dym2:2nd}
\end{eqnarray}
where $r_n^m, Z_n^m \in \mathbb{C}$, $a_n, b_m \in \mathbb{R}$, 
$m,n\in \mathbb{Z}$. 
This is transformed into 
\begin{eqnarray}
&& b_m\tau^*{}_{n}^{m+1}\tau_{n+1}^{m} - a_n\tau^*{}_{n+1}^{m}\tau_{n}^{m+1}
+ (a_n-b_m)\tau^*{}_{n+1}^{m+1}\tau_n^{m}=0\,,\\
&& b_m\tau_{n}^{m+1}\tau^*{}_{n+1}^{m} - a_n\tau_{n+1}^{m}\tau^*{}_{n}^{m+1}
+ (a_n-b_m)\tau_{n+1}^{m+1}\tau^*{}_n^{m}=0\,, 
\end{eqnarray}
via
\begin{equation}
r_n^m=e^{{\rm i}\theta_n^m}=\left(\frac{\tau_n^m}{{\tau^*}_n^m}\right)^2\,,\quad
 \theta_n=\frac{2}{{\rm i}}\log \frac{\tau_n^m}{{\tau^*}_n^m}\,. 
\end{equation}
The $N$-soliton solution is given by 
\begin{eqnarray}
&&\tau_n^m={\rm det}\left(f_{j-1}^{(i)}\right)_{1\leq i,j\leq N}\,,
\quad {\tau^*}_n^m={\rm det}\left(f_{j}^{(i)}\right)_{1\leq i,j\leq N}\,,
\\
&&f_j^{(i)}=\alpha_ip_i^j\prod_{n'}^{n-1}(1-a_{n'}p_i)^{-1}\prod_{m'}^{m-1}(1-b_{m'}p_i)^{-1}
\nonumber\\
&& \quad \qquad 
+\beta_i(-p_i)^j\prod_{n'}^{n-1}(1+a_{n'}p_i)^{-1}\prod_{m'}^{m-1}(1+b_{m'}p_i)^{-1}\,,
\end{eqnarray}
where $p_i,\alpha_i\in \mathbb{R}$, $\beta_i \in \sqrt{-1}\mathbb{R}$
      for $i=1,\cdots, N$. \\
The $M$-breather solution is given by 
\begin{eqnarray}
&&\tau_n^m={\rm det}\left(f_{j-1}^{(i)}\right)_{1\leq i,j\leq N}\,,
\quad {\tau^*}_n^m={\rm det}\left(f_{j}^{(i)}\right)_{1\leq i,j\leq N}\,,
\\
&&f_j^{(i)}=\alpha_ip_i^j\prod_{n'}^{n-1}(1-a_{n'}p_i)^{-1}\prod_{m'}^{m-1}(1-b_{m'}p_i)^{-1}
\nonumber\\
&& \quad \qquad 
+\beta_i(-p_i)^j\prod_{n'}^{n-1}(1+a_{n'}p_i)^{-1}\prod_{m'}^{m-1}(1+b_{m'}p_i)^{-1}\,,
\end{eqnarray}
where $N=2M$, $p_i, \alpha_i, \beta_i\in \mathbb{C}$ for $i=1,\cdots,2M$, 
$p_{2k}=p_{2k-1}^*$, $\alpha_{2k}=\alpha_{2k-1}^*$, 
$\beta_{2k}=-\beta_{2k-1}^*$ for $k=1,\cdots, M$. 
\end{theorem}

\begin{proof}
The proof is similar to Theorem \ref{theorem:discrete-complex-Dym}.  
\end{proof}

\begin{remark}
Equation (\ref{discrete-complex-Dym2:2nd}) gives the discrete hodograph 
transformation 
\begin{equation}
Z_n^m=\sum_{j=0}^{n-1}a_j \sqrt{r_{j+1}^mr_j^m}+Z_0^m\,,
\end{equation}
which leads to the hodograph transformation (\ref{hodograph-complex}) in
 the continuous limit. 
\end{remark}

\begin{theorem}
The fully discrete real Dym equation is given by 
\begin{eqnarray}
&&\frac{\sqrt{r_{n+1}^{m+1}}-\sqrt{r_{n}^{m}}}
{\sqrt{r_{n+1}^{m+1}}+\sqrt{r_{n}^{m}}}
=\frac{b_m+a_n}{b_m-a_n}~
\frac{\sqrt{r_{n}^{m+1}}-\sqrt{r_{n+1}^{m}}}
{\sqrt{r_{n}^{m+1}}+\sqrt{r_{n+1}^{m}}}\,,\label{discrete-real-Dym1:2nd}\\
&&X_{n+1}^m-X_{n}^m=a_n\sqrt{r_{n+1}^mr_n^m}\,,\label{discrete-real-Dym2:2nd}
\end{eqnarray}
where $r_n^m, X_n^m, a_n, b_m \in \mathbb{R}$, 
$m,n\in \mathbb{Z}$. 
This is transformed into 
\begin{eqnarray}
&& b_m\tilde{\tau}_{n}^{m+1}\tau_{n+1}^{m} - a_n\tilde{\tau}_{n+1}^{m}\tau_{n}^{m+1}
+ (a_n-b_m)\tilde{\tau}_{n+1}^{m+1}\tau_n^{m}=0\,,\\
&& b_m\tau_{n}^{m+1}\tilde{\tau}_{n+1}^{m} - a_n\tau_{n+1}^{m}\tilde{\tau}_{n}^{m+1}
+ (a_n-b_m)\tau_{n+1}^{m+1}\tilde{\tau}_n^{m}=0\,, 
\end{eqnarray}
via
\begin{equation}
r_n^m=e^{\theta_n^m}=\left(\frac{\tau_n^m}{\tilde{\tau}_n^m}\right)^2\,,\quad
 \theta_n=2\log \frac{\tau_n^m}{\tilde{\tau}_n^m}\,. 
\end{equation}
The $N$-cusped soliton solution is given by
\begin{eqnarray}
&&\tau_n^m={\rm det}\left(f_{j-1}^{(i)}\right)_{1\leq i,j\leq N}\,,
\quad \tilde{\tau}_n^m={\rm det}\left(f_{j}^{(i)}\right)_{1\leq i,j\leq N}\,,
\\
&&f_j^{(i)}=\alpha_ip_i^j\prod_{n'}^{n-1}(1-a_{n'}p_i)^{-1}\prod_{m'}^{m-1}(1-b_{m'}p_i)^{-1}
\nonumber\\
&& \quad \qquad 
+\beta_i(-p_i)^j\prod_{n'}^{n-1}(1+a_{n'}p_i)^{-1}\prod_{m'}^{m-1}(1+b_{m'}p_i)^{-1}\,,
\end{eqnarray}
where $p_i,\alpha_i, \beta_i\in \mathbb{R}$ for $i=1,\cdots, N$. 
\end{theorem}

\begin{proof}
The proof is similar to Theorem \ref{theorem:discrete-real-Dym}.  
\end{proof}

\begin{remark}
Equation (\ref{discrete-real-Dym2:2nd}) gives the discrete hodograph 
transformation 
\begin{equation}
X_n^m=\sum_{j=0}^{n-1}a_j \sqrt{r_{j+1}^mr_j^m}+X_0^m\,,
\end{equation}
which leads to the hodograph transformation (\ref{hodograph-complex}) in
 the continuous limit. 
\end{remark}


\end{document}